\documentclass[10pt,journal,compsoc]{IEEEtran}

\newtheorem{definition}{Definition}
\newtheorem{strategy}{Strategy} 
\newtheorem{theorem}{Theorem}
\newenvironment{proof}{\begin{IEEEproof}}{\end{IEEEproof}}

\usepackage{multirow}
\usepackage{amsmath}

\usepackage{array}
\newtheorem{upper bound}{Upper bound}

\usepackage{url}
\usepackage{soul}
\usepackage{color, xcolor} 
\usepackage{wrapfig}

\ifCLASSINFOpdf
   \usepackage[pdftex]{graphicx}
\else
\fi

\usepackage[linesnumbered,ruled]{algorithm2e} 

\usepackage[caption=false,font=footnotesize,labelfont=sf,textfont=sf]{subfig}
  \usepackage[caption=false,font=footnotesize]{subfig}
\usepackage{enumitem}

\newcommand\MYhyperrefoptions{bookmarks=true,bookmarksnumbered=true,
	pdfpagemode={UseOutlines},plainpages=false,pdfpagelabels=true,
	colorlinks=true,linkcolor={blue},citecolor={blue},urlcolor={blue},
	pdftitle={High-utility Sequential Rule Mining Utilizing Segmentation Guided by Confidence},
	pdfsubject={Typesetting},
	pdfauthor={ChunKai Zhang},
	pdfkeywords={Data mining, sequential rule, utility mining, segmentation, confidence}}
\ifCLASSINFOpdf
\usepackage[\MYhyperrefoptions,pdftex]{hyperref}
\else
\usepackage[\MYhyperrefoptions,breaklinks=true,urlcolor=black,dvips]{hyperref}

\usepackage{breakurl}
\fi

\begin{document}

\title{High-utility Sequential Rule Mining Utilizing Segmentation Guided by Confidence}

\author{Chunkai Zhang, Jiarui Deng, Maohua Lyu, Wensheng Gan*, and Philip S. Yu,~\IEEEmembership{Life Fellow,~IEEE}
	
	\IEEEcompsocitemizethanks{\IEEEcompsocthanksitem Chunkai Zhang, Jiarui Deng, and Maohua Lyu are with the School of Computer Science and Technology, Harbin Institute of Technology (Shenzhen), Shenzhen 518055, China. (E-mail: ckzhang@hit.edu.cn, 23s151067@stu.hit.edu.cn, 22s151144@stu.hit.edu.cn)

	\IEEEcompsocthanksitem Wensheng Gan is with the College of Cyber Security, Jinan University, Guangzhou 510632, China. (E-mail: wsgan001@gmail.com) 
	
	\IEEEcompsocthanksitem Philip S. Yu is with the Department of Computer Science, University of Illinois at Chicago, IL, USA. (E-mail: psyu@uic.edu)
	
    \IEEEcompsocthanksitem Corresponding author: Wensheng Gan}
}

\IEEEtitleabstractindextext{%

\begin{abstract}
    Within the domain of data mining, one critical objective is the discovery of sequential rules with high utility. The goal is to discover sequential rules that exhibit both high utility and strong confidence, which are valuable in real-world applications. However, existing high-utility sequential rule mining algorithms suffer from redundant utility computations, as different rules may consist of the same sequence of items. When these items can form multiple distinct rules, additional utility calculations are required. To address this issue, this study proposes a sequential rule mining algorithm that utilizes segmentation guided by confidence (RSC), which employs confidence-guided segmentation to reduce redundant utility computation. It adopts a method that precomputes the confidence of segmented rules by leveraging the support of candidate subsequences in advance. Once the segmentation point is determined, all rules with different antecedents and consequents are generated simultaneously. RSC uses a utility-linked table to accelerate candidate sequence generation and introduces a stricter utility upper bound, called the reduced remaining utility of a sequence, to address sequences with duplicate items. Finally, the proposed RSC method was evaluated on multiple datasets, and the results demonstrate improvements over state-of-the-art approaches.
\end{abstract}

}

\maketitle

\IEEEdisplaynontitleabstractindextext

\IEEEpeerreviewmaketitle

\IEEEraisesectionheading{
\section{Introduction}}
\label{sec:introduction}

Sequential data, often referred to as ordered data, consists of observations organized in a fixed order, such as text, DNA sequences, user actions, and log data. Sequential rule mining (SRM) \cite{ 10909343,li2023mcor} can capture predictive or causal relationships between ordered events or items by discovering rules from sequential data. SRM is widely used in the real world, including symptom prediction and prognosis analysis in personalized cancer treatment \cite{floricel2024roses}, cross-domain book recommendation \cite{anwar2022cd},  personalized sentiment analysis for users \cite{skenduli2021mining}, user category prediction \cite{matrouk2023mapreduce}, detection of potential violations in Dockerfiles \cite{2023DRIVE} and discovering treatment patterns in Traditional Chinese Medicine \cite{2024CTM1,2024CTM2}. High-utility sequential rule mining (HUSRM) \cite{zida2015efficient}  can be viewed as an advancement of SRM by integrating utility measures, such as profit, cost, or benefit, into rule evaluation. This enables the identification of rules that are valuable in real-world scenarios. For example, in recommendation tasks, HUSRM incorporates user rating data as utility to discover which movies users prefer and to recommend movies based on individual rating behaviors \cite{jiang2024novel}. In bioinformatics, it uses levels of gene expression as a utility to identify genes that significantly influence phenotypic traits, thereby facilitating human gene expression analysis \cite{SEGURADELGADO2022116411}. In the domain of cybersecurity, the severity levels of log records are used as a utility to identify which log entries are associated with abnormal behaviors, supporting effective anomaly detection \cite{2023Anomaly}. In healthcare, the hazard ratios of medical examination features can be used as utility to extract rules related to cardiovascular diseases, aiding in prevention and early diagnosis \cite{IQBAL2022102347}. In e-commerce, product profits are treated as utility to mine personalized recommendation rules that both reflect consumer preferences and maximize merchant revenue \cite{zida2015efficient}.

The utility of the rule is an important metric for evaluating its value. Therefore, the utility needs to be calculated for each generated rule across the candidate space, making the process computationally intensive and complex. Since utility does not possess the downward closure property, it cannot be directly applicable to prune the search space. To overcome this limitation, the notion of an upper bound on utility has been introduced, which satisfies the closure property and allows effective pruning of the candidate space. \cite{zida2015efficient} first defined the HUSRM task and introduced an upper bound for the rule's utility, called the estimated sequence utility (SEU). To accelerate rule generation, they proposed the utility table data structure and designed a bitmap to compute rule occurrence counts efficiently. Huang et al. \cite{huang2023us} addressed the issue of loose utility upper bounds and low pruning efficiency by proposing new bounds, including Left Expansion Estimated Utility (LEEU) along with its right-expansion counterparts, to enhance pruning performance. TotalSR \cite{zhang2024totally} is a totally ordered mining approach, which can discover rules from sequences where the order of events is strictly emphasized. It further tightened the utility upper bound by proposing a new upper bound and designed an auxiliary antecedent record table to track rule confidence. USER \cite{lin2023user} targeted the limitation of existing methods that overlook repetitive items by partitioning the utility computation during the expansion process, thereby enabling effective mining on databases containing repetitive items. In addition, some studies are focusing on other variants of the HUSRM, such as mining high average-utility rules \cite{SEGURADELGADO2022116411} and discovering rare high-utility sequential rules \cite{IQBAL2022102347}.

Although the above HUSRM algorithms have achieved significant improvements in reducing the search space, they have made little progress from the perspective of rule generation. All of them utilize a left or right expansion (LRE) strategy to generate rules, extending the antecedent or consequent of existing rules from either the left or the right at the expansion position. The LRE strategy was first proposed by Zida et al. \cite{zida2015efficient}, who observed that rules could be repeatedly derived through various combinations of left and right extensions, and this strategy restricts the expansion direction to avoid redundancy of rule generation and ensure the completeness of discovered rules. However, we observe the situation that this rule generation strategy leads to redundant utility computations, as multiple distinct rules may be composed of the same set of items forming an identical sequence. For example, for  the existing rules $\{a$ $\rightarrow$ $d\}$ and $\{a$ $\rightarrow$ $c\}$ in the sequence $<$$a$, $b$, $c$, $d$$>$, the extended rules are $\{abc$ $\rightarrow$ $d\}$ and $\{ab$ $\rightarrow$ $cd\}$, respectively. The utility computation process for both rules is essentially the same, since the items $a$, $b$, $c$, and $d$ involved in these rules are identical, leading to redundant utility computations. The issue of redundant utility computation becomes more pronounced when sequences contain repetitive items \cite{lin2023user}, as the utility must be calculated multiple times to determine the maximum utility of the expanded rule. For instance, consider the sequence $<$$a(1)$, $b(3)$, $c(4)$, $d(2)$, $a(1)$, $a(2)$, $b(1)$, $c(3)$, $d(3)$$>$, where the number in parentheses indicates the utility of each item. The maximum utility of the item $a$ is 2 (at position 6), and the maximum utility of items 
$ab$ is 4 (with $a$ at position 1 and $b$ at position 2). After expansion, the selected position of $a$ has changed, meaning the utility of all possible occurrences of $ab$ must be evaluated during expansion. Take the rules $\{a$ $\rightarrow$ $d\}$ and $\{a$ $\rightarrow$ $c\}$ as examples. After expansion, they generate $\{ac$ $\rightarrow$ $d\}$ and $\{a$ $\rightarrow$ $cd\}$, respectively. Since the items $<$$a$, $d$$>$ and $<$$a$, $c$$>$ each appear four times in sequence, four utility computations are required for each rule before expansion. As a result, the utility of $<$$a$, $c$,  $d$$>$ is computed one extra time. In this example, utility values are recalculated for all four occurrences of this sequence after expansion.

Since rule generation methods based on the LRE algorithm initiate mining from specific initial rules, the candidate rules generated by each expansion are different, so it is necessary to calculate the utility of candidate rules every time. The redundant utility computations incurred in this process cannot be avoided by methods such as simply storing the utilities of the discovered rules; therefore, new methods are needed to address the problem of redundant utility computations. A novel high-utility sequential rule generation method is proposed that derives all corresponding rules from one sequence. This approach requires computing the utility of the sequence only once, allowing the obtained utility to be shared across all derived rules. As a result, redundant utility computations are eliminated, thereby improving overall efficiency. However, this idea faces three challenges. First, although utility upper bounds are capable of efficiently producing high-utility sequences, these sequences do not necessarily yield high-utility sequential rules. Second, how to efficiently derive rules from existing sequences remains an open question. Third, it is uncertain whether the new rule generation method could ensure the completeness of the discovered rules. In this study, we present a sequential Rule mining algorithm that utilizes Segmentation guided by Confidence, named  RSC, which addresses the redundant utility computation caused by LRE. It proposes a segmentation approach that precomputes the confidence of the segmented rules by using the support of candidate subsequences. Once the segmentation point is determined, all rules with different antecedents and consequents are generated at once. Since all rules are generated from the same sequence with identical utility, redundant computations are effectively avoided. As a segmentation-based method, RSC designs a data structure called the utility-linked table, which uses pointers to index the projected database, thereby accelerating candidate generation. RSC also proposes a fast segmentation strategy to accelerate rule generation. To handle sequences with repetitive items, RSC introduces a tighter utility upper bound called Reduced Remaining Utility (RRU), together with associated pruning techniques designed to narrow the search space. Through theoretical analysis, the segmentation-based method is shown to generate complete rules. This paper makes the following key contributions:

\begin{itemize}
    \item To address the issue of redundant utility computation caused by LRE-based rule generation, we design a sequential rule mining algorithm that utilizes segmentation guided by confidence. 
    
    \item We propose a new structure, the utility-linked table, to project the database and enable efficient high-utility sequence generation. A new utility upper bound called reduced remaining utility, along with corresponding pruning strategies, is introduced to reduce the search space.
    
    \item Evaluations conducted on both real-world and synthetic datasets demonstrate that the proposed algorithm reduces runtime by approximately 50$\%$, showing its high efficiency.
\end{itemize}

This paper is structured as follows. HUSRM is concisely reviewed in Section \ref{sec:relatedwork}. Section \ref{sec:preliminaries} presents the basic concepts and problem definition. Section \ref{sec:method} details the RSC algorithm and its procedures. Section \ref{sec:experiments} details the experimental setup and performance evaluation using both real-world and synthetic datasets. Finally, Section \ref{sec:conclusion} summarizes the main findings of this study and outlines future work.

\section{Related work}  \label{sec:relatedwork}

This section reviews existing research on SRM, high-utility sequential pattern mining (HUSPM), and HUSRM.

\subsection{Sequential rule mining }

Combinations of items that appear more frequently than a given support threshold are typically identified through pattern mining techniques \cite{380415,SUN2024112398,LAN2025104077,ONP}, where patterns with higher frequencies are typically considered more valuable. However, it primarily focuses on identifying frequent combinations of items and lacks the capability to predict the likelihood of subsequent items based on current observations. Therefore, sequential rule mining (SRM) was proposed \cite{zaki2001spade,10008085}, which aims to mine all rules expressed as $X$ $\rightarrow$ $Y$, in which $X$ serves as the antecedent and $Y$ as the consequent. It introduces the concept of confidence and focuses more on discovering associations between patterns. Formally, confidence is calculated as the rule’s support over the support of the antecedent. The confidence metric is used to represent the likelihood that one pattern will appear given the occurrence of another, making it suitable for predictive tasks. SRM is divided into two categories: totally ordered \cite{lo2009non,pham2014efficient,li2023mcor} and partially ordered \cite{fournier2011rulegrowth,fournier2014erminer,fournier2015mining}. A totally ordered rule requires that the items in $X$ and $Y$ follow the exact order found in the original database, while a partially ordered rule only requires that $X$ and $Y$ are relatively ordered, which means that there are no specific requirements for the order of items within $X$ and $Y$. CMDeo \cite{fournier2012cmrules} first discovers all association rules and then uses them to produce sequential rules. CMRules \cite{fournier2012cmrules} first introduces the concept of LRE. It expands association rules to the left or right to obtain sequential rules. Subsequently, RuleGrowth \cite{fournier2011rulegrowth} and TRuleGrowth \cite{fournier2015mining} were proposed. Both of them use the LRE method to directly obtain sequential rules without first generating association rules. They are also partially ordered because the authors argue that some totally ordered rules may be overly specific and thus rarely occur in practice \cite{fournier2015mining}. In ERMiner \cite{fournier2014erminer}, rule generation is expedited by organizing rules sharing identical antecedent or consequent, along with the use of a sparse count matrix to facilitate pruning. In addition, numerous studies have focused on variants of SRM. TaSRM \cite{10909343} focuses on mining user-specified and relevant sequential rules for specific targets. To mine maximal co-occurrence non-overlapping sequential rules, MCoR-Miner \cite{li2023mcor} employs frequent item-based techniques and a binomial enumeration tree approach. CoUSR \cite{10606108} integrated the concept of correlation into HUSRM to discover rules with strong correlations.

\subsection{High-utility sequential pattern mining}

Developed as an extension of frequent pattern mining \cite{agrawal1996fast}, high-utility pattern mining (HUPM) \cite{liu2012direct,li2025hupsp,chen2024towards} was proposed. Rather than focusing solely on item frequency, HUPM emphasizes the actual utility of patterns, which reflects user interests, and its value is often assessed through metrics such as profit, cost, or benefit. In this framework, each item is linked to a corresponding utility score, and the goal is to identify all patterns with a total utility surpassing a predefined minimum threshold. To reduce the search space, an upper bound sequence weighted utility (SWU) was proposed. SWU satisfies the downward closure property and is the first customized upper bound proposed for utility-based search space pruning. Subsequently, Yin et al. \cite{yin2012uspan} were the first to define HUSPM and proposed the USpan algorithm. By utilizing a utility matrix structure, USpan can store the order of item occurrences in sequences and introduce sequence projected utilization (SPU) for pruning. However, both SWU and SPU produce overly loose upper bounds, leading to significant gaps between their estimates and the actual utilities of patterns. Moreover, SPU is not entirely accurate, as it may mistakenly prune some actual high-utility sequential patterns. To shrink the search space and make the algorithm more efficient, \cite{alkan2015crom} proposed a data structure and pruning strategy based on Cumulated Rest of Match (CRoM). They also used the utility of parent patterns to calculate the child patterns' utility. HUS-Span \cite{wang2016efficiently} proposed two new upper bounds on utility—prefix extension utility and reduced sequence utility. Additionally, ProUM  \cite{gan2020proum} introduced a projection strategy and a utility array to record utility values. It reduces memory usage by projecting the database. HUSP-ULL \cite{gan2020fast} further proposed the UL-list structure, which efficiently creates projected databases based on prefix sequences. HUSP-SP \cite{zhang2023husp} introduces a more compact data structure called sequence-array to reduce memory usage and introduces a more restrictive utility bound called tighter reduced sequence utility.

In addition to improving existing algorithms, there are also many studies on other variants of HUSPM algorithms. TUSQ \cite{zhang2022tusq} focuses on mining patterns that meet certain requirements instead of considering all items.  Zhang et al. \cite{zhang2023mining} considered that sequences may contain items with negative utility.  U-HPAUSM \cite{duong2025u} focuses on the issue that general mining may return long patterns. It discovers high average-utility patterns within uncertain quantitative databases.

\subsection{High-utility sequential rule mining}

As an extension of SRM, high-utility sequential rule mining (HUSRM)  integrates the notion of utility. \cite{zida2015efficient} was the earliest to propose HUSRM, which mines partially ordered rules. They used the LRE strategy to obtain rules and used bitmaps to calculate the support of $X$. US-Rule \cite{huang2023us} considered the characteristics of LRE and proposed four new utility upper bounds, which greatly shrink the search space. In addition, \cite{lin2023user} proposed the USER algorithm, which improves existing methods to support mining on datasets containing repetitive items. TotalSR \cite{zhang2024totally} proposed a totally ordered HUSRM algorithm. As noted in TotalSR, totally ordered algorithms cannot use bitmap structures to compute the support of $X$, unlike partially ordered algorithms. Moreover, considering sequential order increases the search space, placing greater performance demands on utility upper bounds and pruning strategies. Therefore, TotalSR used the auxiliary antecedent record table to record data information. It also introduced two new utility upper bounds and employed confidence for pruning to improve algorithm efficiency. As noted in TotalSR, partially ordered mining may filter out infrequent but valuable rules and generate false rules \cite{zhang2024totally}. This is especially evident in data that places great emphasis on the order of events, such as stock trading data and clickstream data. 

Until now, numerous applications and algorithms have explored high-utility rule mining. However, these methods all adopt the LRE strategy for rule generation, causing redundant utility computations. Thus, this paper proposes the RSC algorithm, which avoids redundant utility calculations for rules with identical item sequences and employs the segmentation strategy to generate all sequence rules at once, improving algorithm efficiency.

\section{Background}   \label{sec:preliminaries}
This section explains key definitions related to this paper and then states the HUSRM problem.

\subsection{Notations and definitions}

\begin{definition}[Sequence and sequence database \cite{zaki2001spade}]
  \rm  Consider a set $\Sigma$ $=$ \{$i_1$, $i_2$, $\cdots$, $i_m$\} consisting of distinct items. A sequence $s$ is an ordered collection of elements, represented as $s$ $=$ $<$$e_1$, $e_2$, $\cdots$, $e_l$$>$, where $e_k$ $\in$ $\Sigma$ $(1$ $\le$ $k$ $\le$ $l)$. A sequence database $\mathcal{D}$ is defined as a collection of sequences, denoted by $<$$s_1$, $s_2$, $\cdots$, $s_n$$>$, where $s_i$ $(1$ $\le$ $i$ $\le$ $n)$. Note that sequences do not impose any restrictions on the occurrence count of individual items in a sequence; therefore, multiple identical items may occur in a sequence $s$. Typically, each distinct item within the database is assigned a positive value known as the external utility, whereas the item’s presence within a sequence is associated with an internal utility. The true utility of an item within a sequence is calculated by multiplying its external utility with its internal utility, following the HUSRM framework. In this study, this step is omitted, and the actual utility of the $e_k$ in a particular $s$ is represented as $u(e_k,s)$.
\end{definition}

 It is assumed that no events occur simultaneously within a sequence, meaning that itemsets are not present. This assumption is consistent with the majority of real-world sequence data. Table \ref{database} shows an example sequence database with five sequences.

\begin{table}[h]
    \centering
    \caption{Sequence database}
    \label{database}
    \begin{tabular}{c|c}
    \hline
    \textbf{SID} & \textbf{Sequences} \\
    \hline
    $s_1$ & $\langle (a, 1), (c, 2), (c, 3) \rangle$ \\
    \hline
    $s_2$ & $\langle (b, 5), (c, 2), (e, 8), (b, 6) \rangle$ \\
    \hline
    $s_3$ & $\langle (a, 2), (c, 2), (f, 10) \rangle$ \\
    \hline
    $s_4$ & $\langle (a, 2), (a, 1), (a, 2), (b, 6), (c, 3), (a, 3)\rangle$ \\
    \hline
    $s_5$ & $\langle (d, 1), (a,1), (b, 4) \rangle$ \\
    \hline
    \end{tabular}
\end{table}

 \begin{definition}[Sequence occurrence and utility \cite{zida2015efficient}]
 \label{sequenceUtility}
  \rm  For two sequences $s^{\prime}$ $=$ $<$$e_1$, $\cdots$, $e_n$$>$ and  $s$ $=$ $<$$e_1^{\prime}$, $\cdots$, $e_{m}^{\prime}$$>$, we say that $s^{\prime}$ appears in $s$ if and only if there exist indices 1 $\le$ $j_1$ $<$ $\cdots$ $<$ $j_n$ $\le$ $m$  such that $e_1$ $=$ $e_{j_1}^{\prime}$, $\cdots$, $e_n$ $=$ $e_{j_n}^{\prime}$. Assume that sequence $s^{\prime}$ appears multiple times in sequence $s$, denoted as $O(s^{\prime},s)$ $=$ $\{$$O_1, O_2$, $\cdots$, $O_m$$\}$, where each $O_i$ denotes the index set of a distinct occurrence of $s^{\prime}$ in $s$. Furthermore, we define $seq(s^{\prime})$ $=$ $\{$$s$ $|$ $O(s^{\prime},s)$ $\neq$ $\emptyset$$\}$ to represent the sequences that contain $s^{\prime}$, and $sup(s^{\prime})$ $=$ $|seq(s^{\prime})|$ represents the support of $s^{\prime}$. We represent the utility of $s^{\prime}$ in $s$ as $u(s^{\prime},s)$ $=$ max$\{u(O_i)$ $|$ $O_i \in O(s^{\prime},s)\}$, where $u(O_i)$ $=$ $\sum_{k \in O_i} u(e_k,s)$. The overall utility of $s^{\prime}$ within $\mathcal{D}$ is represented by $u(s^{\prime})$ = $\sum_{s_i\in \mathcal{D}}u(s^{\prime},s_i)$. Note that if the sequence $s^{\prime}$ contains just one item, the above can be used to represent the properties of that specific item. Note that the position statistics in this paper are based on the total order. Therefore, the rules mined are also totally ordered \cite{zhang2024totally}.
\end{definition}

For example, the sequence $s$ $=$ $<$$b$, $c$$>$ appears in $s_2$ and $s_4$, $seq(s)$ $=$ $\{$$s_2$, $s_4$$\}$, $u(s)$ $=$ $u(s,s_2)$ $+$ $u(s,s_4)$ $=$ $(5+2)$ $+$ $(6+3)$  $=$ $7$ + $9$ = 16.

\begin{definition}[Sequential rule \cite{zaki2001spade}]
    \rm   A sequential rule (SR),  denoted as $r$ $=$ $X$ $\rightarrow$ $Y$,  represents a dependency between two sequences, with $X$ being the antecedent and $Y$ the consequent. Specifically, sequential rules do not allow repeated items and must satisfy $X$ $\cap$ $Y$ = $\emptyset$  and $X$ $\neq$ $\emptyset$, $Y$ $\neq$ $\emptyset$. An SR $r$ can be interpreted as: if the items in $X$ occur, then the items in $Y$ are likely to follow. 
\end{definition}

\begin{definition}[Sequential rule occurrence and utility \cite{zida2015efficient}]
   \rm For an SR $r$ = $\{$$e_1$, $\cdots$, $e_{k}$$\}$ $\rightarrow$ $\{$$e_{k+1}$, $\cdots$, $e_{n}$$\}$ and the sequence $s$ $=$ $<$$e_1^{\prime}$, $\cdots$, $e_{m}^{\prime}$$>$, we say that $r$ appears in $s$ if and only if the rule sequence $s_r$ $=$ $<$$e_1$, $\cdots$, $e_{k}$, $e_{k+1}$, $\cdots$, $e_{n}$$>$ appears in $s$. Suppose that rule $r$ appears multiple times within $s$, represented as $O(r,s)$ $=$ $\{$$O_1, O_2$, $\cdots$, $O_m$$\}$, where $O_i$ represents the set of indices corresponding to occurrences of $r$ in $s$. $r$'s utility in sequence $s$ is indicated as $u(r,s)$ $=$ max$\{u(O_i)$ $|$ $O_i \in O(r,s)\}$, where $u(O_i)$ $=$ $\sum_{n \in O_i} u(e_n,s)$. The sum utility of $r$ in $\mathcal{D}$ is represented as $u(r)$ = $\sum_{s_i\in \mathcal{D}}u(r,s_i)$.
\end{definition}

Take sequences $s_1$, $s_3$, $s_4$ and $r$ $=$ $\{a \rightarrow c\}$ as an example. The $r$ appears in $s_1$, $s_3$ and $s_4$. Specifically, $r$ appears three times in $s_4$, denoted as $O(r,s_4)$ $=$ $\{$$O_1, O_2, O_3\}$, where $O_1 = \{1,5\}$, $O_2 = \{2,5\}$, $O_3 = \{3,5\}$, respectively. The total utility of $r$ $=$ $\{a \rightarrow c\}$ is $u(r)$ $=$ $u(r,s_1)$ $+$ $u(r,s_3)$ $+$ $u(r,s_4)$ $=$ $(1+2)$ $+$ $(2+2)$ $+$ max$\{$$(2+3)$, $(1+3)$, $(2+3)$$\}$ $=$ $7$ + $5$ = 12.

\begin{definition}[Sequential rule support and confidence \cite{zaki2001spade}]
    \label{supAndConf}
     Let $r$ $=$ $X$ $\rightarrow$ $Y$ be a sequential rule, $\mathcal{D}$ be a database and $s$ be a seuqence, $s\in\mathcal{D}$. If $O(r,s) \neq \emptyset$, then we say that $r$ occurs in $s$. Let $seq(r)$ $=$ $\{$$s$ $|$ $O(r,s)$ $\neq$ $\emptyset$$\}$ denote the $s$ in which rule $r$ appears, and $ant(r)$  as the set of sequences containing its antecedent. The value of $sup(r)$ $=$ $|seq(r)|$ represents the support of $r$. We define the confidence of rule $r$ as $conf(r)$ $=$ $sup(r)/sup(X)$ $=$ $|seq(r)|/|ant(r)|$.
\end{definition}

For instance, as shown in Table \ref{database}, the confidence of $r$ $=$ $\{a \rightarrow c\}$ is $conf(r)$ $=$ $|seq(r)|/|ant(r)|$ $=$ $|\{$$s_1$,  $s_3$, $s_4$ $\}|$ $/$ $|\{$$s_1$,  $s_3$, $s_4$, $s_5$ $\}|$ $=$ $3/4$ $=$ 0.75.

\subsection{Problem statement}

This subsection provides the formal problem definition and presents an example to illustrate the process.

\begin{definition}[High-utility sequential rule \cite{zida2015efficient}]
    \label{HUSRDefinition}
    \rm A rule $r$ is defined as a high-utility sequential rule (HUSR) if it meets the conditions $u(r)$ $\geq$ \textit{minutil} and \textit{conf(r)}$\geq$ \textit{minconf}, where \textit{minutil} is a user-specified utility threshold and \textit{minconf} is a confidence threshold in the range $($0$,$1$)$.
\end{definition}

\textbf{Problem statement:} Let $\mathcal{D}$ be a sequence database, with  \textit{minutil} = $\delta$ $\times$ $u(D)$ and  \textit{minconf} given, according to the definition of HUSR, the problem is to find all totally ordered HUSRs in $\mathcal{D}$ that satisfy the \textit{minutil} and \textit{minconf} requirements. Table \ref{ruleExample} presents the output rules with \textit{minutil} = $u(\mathcal{D})$ $\times$ $\delta$ $=$ $64$ $\times$ $0.1$ $=$ 6.4  and \textit{minconf} = 0.6.  For $r$ $=$ $\{a$ $\rightarrow$ $b\}$, $u(r)$ $=$ $13$ $>$ \textit{minutil}, $conf(r)$ $=$ $2$ $/$ $4$ $=$ $0.5$ $<$ \textit{minconf}, therefore, $r$ is not a HUSR. 

\begin{table}[h]
	\centering
	\caption{The discovery rules with $\delta$ $=$ 0.1, \textit{minutil} = 6.4, and \textit{minconf} = 0.6.}
	\label{ruleExample}
        \begin{tabular}{|c|c|c|c|}  
    		\hline 
    		\textbf{ID} & \textbf{HUSR} & \textbf{utility} & \textbf{conf} \\
    		\hline 
    		$r_1$ & \{a\} $\rightarrow$ \{c\} & 13 & 0.75 \\ \hline
                $r_2$ & \{b\} $\rightarrow$ \{c\} & 16 & 0.67 \\ \hline
    		$r_3$ & \{c, e\} $\rightarrow$ \{b\} & 16 & 1.00 \\ \hline 
                $r_4$ & \{e\} $\rightarrow$ \{b\} & 14 & 1.00     \\ \hline 
            
    \end{tabular}	
\end{table}

\section{RSC algorithm}   \label{sec:method}

\begin{figure*}[h]
    \centering
    \includegraphics[width=1\textwidth]{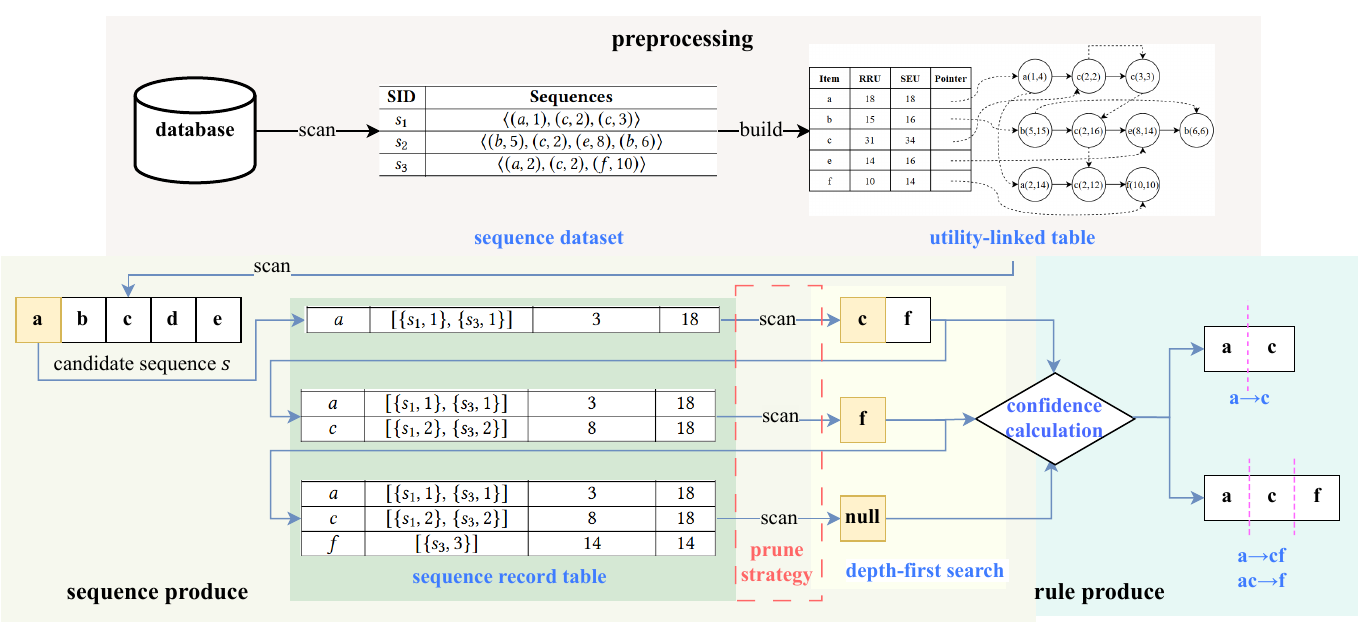}
    \caption{The flowchart of the RSC algorithm: $($$1$$)$ The utility-linked table stores sequence data and facilitates depth-first search; $($$2$$)$The sequence record table stores candidate sequence information and leverages utility to facilitate pruning; $($$3$$)$In rule production, information from the sequence record table is utilized to identify cut points and generate all rules for a sequence in one pass.}
    \label{flowchart}
\end{figure*}

In this section, upper bounds, data structures, some critical rule-producing methods, and pruning strategies used throughout this work are introduced. The RSC algorithm’s flowchart is depicted in Fig. \ref{flowchart}, which illustrates the general process of extracting HUSRs from the given example database. The sequences in the database are preprocessed and stored in the utility-linked table. The initial candidate sequence s is obtained by scanning the header table of the utility-linked table. Subsequently, candidates are taken in order. The corresponding sequence record table is then constructed by scanning backward from the projected position. Meanwhile, any new candidates discovered during this process are added to the depth-first search candidate queue. The sequence record table is completed using a depth-first search approach. Finally, after determining the slicing position, the slice method is applied to the data in the sequence record table to output all the rules that can be generated from the current sequence. As can be seen from the entire workflow, the utility-linked table enables depth-first search to proceed continuously without backtracking. Moreover, each sequence is recorded only once in the sequence record table. After confirming the cut-off points, all rules corresponding to the sequence are generated in a single pass, eliminating the need for repeated utility searches for the same sequence as required by the LRE strategy.

\subsection{Upper bounds}

In utility-driven mining, utility upper bounds are often used to reduce the search space. The idea of remaining utility has been extended into different bounds in various HUSRM algorithms. Since RSC does not adopt the LRE strategy, most previous utility upper bounds cannot be directly applied. Therefore,  new utility upper bounds are proposed for pruning, while retaining the SEU upper bound introduced in HUSRM \cite{zida2015efficient}.

\begin{upper bound}[sequence estimated utility of an item \rm \cite{zida2015efficient}]
    \rm Given an item $i$ and a sequential database $\mathcal{D}$. The sequence estimated utility of item $i$, denoted as \textit{SEU}$(i)$ $=$ $\sum_{s \in seq(i)}u(s,s)$. 
    \label{SEU}
\end{upper bound}
    
\begin{upper bound}[remaining utility of an item \rm  \cite{zida2015efficient}]
    \rm For the distinct item $i$ and the sequence $s$ $=$ $<$$e_1$, $\cdots$, $e_k$, $\cdots$, $e_n$$>$, we say that $i$ appears in $s$, $i\in s$, only if $\exists$ $s_k$ $=$ $i$. The remaining utility of item $e_k$, denoted as \textit{RU}$(e_k,s)$ $=$ $\sum_{m=k}^{m=n}u(e_m,s)$.
    \label{RU}
\end{upper bound}

\begin{upper bound}[reduced remaining utility of an item]
   \rm For the distinct item $i$ and the sequence $s$ $=$ $<$$e_1$, $\cdots$, $e_k$, $\cdots$, $e_n$$>$, $i$ $\in$ $s$ and $s_k$ = $i$. The reduced remaining utility of item $e_k$, denoted as \textit{RRU}$(e_k,s)$ $=$ $u(e_k,s)$ $+$ $\sum\limits_{i^{\prime}\in s^{\prime}}max\{u(i^{\prime},s)\}$, where $s^{\prime}$ $=$ $<$$e_{k+1}$, $\cdots$, $e_n$$>$ called the expansion sequence with sequence $s$ and $e_k$, is denoted as \textit{ES}$(e_k,s)$. 
    \label{RRU}
\end{upper bound}

According to the Upper Bound \ref{RRU}, RRU considers both the utility of the current items in the sequence and the utility of the following items. As a result, RU is guaranteed to exceed the actual utility of the sequence. An example is provided in Table \ref{database}, \textit{RU}$(a,s_1)$ $=$ 6, \textit{RRU}$(a,s_1)$ $=$ 4. In $s_4$, $a$ appears multiple times. Taking the first occurrence of $a$ as an example, denoted as $a_0$. \textit{RU}$(a_0,s_4)$ $=$ 17, while \textit{RRU}$(a_0,s_4)$ $=$ 2 $+$ 6 $+$ 3 $=$ 11. Note that \textit{RRU} removes the utility of repeated items, making it tighter than \textit{RU} and more effective for search space pruning. Specifically, when a sequence has no repeated items, \textit{RRU} is equal to \textit{RU}, and \textit{RRU} degenerates to \textit{RU}.

\begin{definition}[expansion of a Sequence with an item]
  \rm  Consider a sequence $s$ $=$ $<$$e_1$, $\cdots$, $e_m$$>$, $i$ is a single item, and let $\mathcal{D}$ denote the sequence database. We say that $s$ can be extended to $s_1$ $=$ $<$$e_1$, $\cdots$, $e_m$, $i$ $>$ by item $i$ only if $\exists$ $s^*$ $\in$ $\mathcal{D}$, $i$ $\in$ $s^*$, $O(s_1,s^*)$ $\ne$ $\emptyset$. All qualified $s^{*}$ and expansion position $e_p$ $=$ $i$  form the set of expansion positions of $s$ with $i$, represented as \textit{SEP}$(s,i)$ $=$ $\{$ $(s^*_1,{e_p}_1)$, $\cdots$, $(s^*_n,{e_p}_n)$ $\}$. 
\end{definition}

\begin{upper bound}[reduced remaining utility of a sequence]
    \label{RRS}
    \rm For the sequence $s$ that can be extended by item $i$, $\mathcal{D}$ is the database. The reduced remaining utility of $s$ with item $i$, represented as \textit{RRS}$(s,i)$ $=$ $\sum_{ s^{*}\in \Pi_1(SEP(s,i)) \land e_p \in s^{*} \land e_p \in \Pi_2(SEP(s,i)) }max\{u(s,s^{*})) + RRU(e_p,s^*)\}$ 
\end{upper bound}

Since Upper Bound \ref{RRS} incorporates the utility of the extended sequence along with that of items appearing later in the database, it results in an RRS value that overestimates the true utility of the extension. For instance, in Table \ref{database}, $<$$a,$ $b$$>$ appears in $s_4$ and $s_5$. To ensure that a sequence can be used to generate rules, the relative order of its items in the original database must be preserved. Therefore, in $s_4$, only $c$ can be an extension item of $s$, while in $s_5$, no extension items exist. \textit{RRS}$(s,c)$ $=$ $u(s,s_4)$ + $RRU(c,s_4)$ $=$ 8 $+$ 6 $=$ 14. For another sequence $s^{\prime}$ = $<$$a$$>$, \textit{RRS}$(s^{\prime},c)$ $=$ max$\{$$u(s^{\prime},s_1) + \textit{RRU}(c,s_1)$$\}$ $+$ $($$u(s^{\prime},s_3) + \textit{RRU}(c,s_3)$$)$ $+$  $($$u(s^{\prime},s_4)$ + $\textit{RRU}(c,s_4)$$)$ $=$ $max\{3,$ $4\}$ $+$ $(2+12)$ $+$ $max\{8,$ $7$, $8\}$ $=$ $26$.

\begin{strategy}[sequence expansion pruning strategy]
    \label{seps}
    \rm If sequence $s$ can be extended by item $i$ to sequence  $s_1$, $u(s_1)$  will not exceed the Upper Bound  \textit{RRS}$(s_1)$. Therefore, if  \textit{RRS}$(s,i)$ $\leq$ \textit{minutil }, then no matter how the sequence is extended, the items it represents, arranged in sequence order, cannot form a HUSRM. Thus, when \textit{RRS}$(s,i)$ $\leq$ \textit{minutil }, RSC can stop the extension of $s_1$.
\end{strategy}

\begin{proof}
    For the sequence $s$ that can execute an expansion with an item $i$, \textit{SEP}$(s,i)$ $=$ $\{$ $(s^*_1,{e_p}_1)$, $\cdots$, $(s^*_n,{e_p}_n)$ $\}$, and  $\mathcal{D}$ is the database. Extending $s$ with item $i$ results in $s_1$. According to the Upper Bound \ref{RRU} and Definition \ref{sequenceUtility}, k is an arbitrary index, $max \big\{ u(s,s^{*}_k)$ $+$ $u({e_p}_k,s^*_k)$ $+$ $u(ES({e_p}_k,s^*_k),s^*_k) \big\}$ $\leq$ $max \big\{u(s,s^{*}_k)$ $+$ $u({e_p}_k,s^*_k)$ $+$ $\sum_{i^{\prime}\in ES({e_p}_k,s^*_k)}max\{u(i^{\prime},s^*)\}\big\}$   $=$  $max\big\{u(s,s^{*}_k)$ $+$ \textit{RRU}$({e_p}_k,s^*_k)\big\}$, since calculating the utility of a sequence requires considering the order of items, while \textit{RRU} does not. This holds for every extendable sequence. Therefore, after summing all extendable sequences, and based on the upper bound \ref{RRS}, the left side of the inequality represents the maximum utility of the sequence $s$ after any possible extensions, while the right side represents \textit{RRS}$(s,i)$. This holds for every extendable sequence. Thus, no matter how $s$ is extended, its utility in $\mathcal{D}$ will not exceed \textit{RRS}$(s,i)$. Therefore, when \textit{RRS}$(s,i)$ $\leq$ \textit{minutil}, RSC can stop the extension of $s_1$.
\end{proof}

\begin{strategy}[early item pruning strategy]
    \label{seus}
    \rm An item $i$ is considered promising only if \textit{SEU(i)} $\geq$ \textit{minutil}. If not, the item $i$ should be removed from the database.
\end{strategy}

\begin{proof}
    According to the Upper Bound \ref{SEU}, \textit{SEU}$(i)$ $=$ $\sum_{s \in seq(i)}u(s,s)$ $=$ $\sum_{s \in seq(i) \land i^{\prime}\in s}max\{u(i^{\prime},s)\}$. Assume that there exists a sequence $s^{\prime}$, $u(s^{\prime})$ $>$ \textit{minutil}, $i$ $\in$ $s^{\prime}$ and \textit{SEU}(i) $<$ \textit{minutil}. $u(s^{\prime})$ $=$ $\sum_{s\in \mathcal{D}}u(s^{\prime},s)$ $=$ $\sum_{s\in \mathcal{D} \land s^{\prime} \subset s} max\{u(O_i)$ $|$ $O_i \in O(s^{\prime},s)\}$ $\leq$ $\sum_{s\in \mathcal{D} \land s^{\prime} \subset s \land i^{\prime}\in s} max\{u(i^{\prime},s)\}$ $\leq$ \textit{SEU}$(i)$ $=$ $\sum_{s\in seq(i)\land i^{\prime}\in s} max\{u(i^{\prime},s)\}$. Therefore, $u(s^{\prime})$ $\leq$ \textit{SEU}$(i)$ $<$ \textit{minutil}, $u(s^{\prime})$ $>$ \textit{minutil}  form a contradiction, meaning that the assumption is not valid. Thus, a sequence $s$ containing an item with \textit{SEU}$(i)$ $<$ \textit{minutil} cannot achieve $u(s)$ $\geq$ \textit{minutil}, regardless of how it is extended.
\end{proof}

\subsection{Data structures}

To discover all HUSRMs, the RSC algorithm needs to scan the database to generate all candidate sequences. Directly scanning of the database is highly inefficient. Therefore, RSC proposes a data structure called the utility-linked table (\textbf{ULT}) and a sequence record table (\textbf{SRT}) to improve the algorithm's efficiency. The utility-linked table (\textbf{ULT}) has two parts:  a header table and a utility list. To simplify the presentation, only the sequences $s_1$, $s_2$, and $s_3$ from table \ref{database} are used for illustration, as illustrated in Fig. \ref{ULTExample}. \textbf{The header table} consists of four fields. The \textit{Item} field stores the name of the item. The \textit{RRU} field stores the total of \textit{RRU} values for the item across all sequences in the database. The \textit{Pointer} field stores a pointer to the first occurrence of the item in the utility list. \textbf{The utility list} consists of nodes and pointers. Each node stores the actual utility of the current item and its \textit{RRU} value within the sequence. Solid pointers indicate the position of the next item in the original database sequence, and if there is no next item, the pointer is set to null. Dashed pointers point to the next occurrence of the same item in the database, prioritizing the item within the same sequence.

\begin{figure}[h]
    \includegraphics[width=0.5\textwidth]{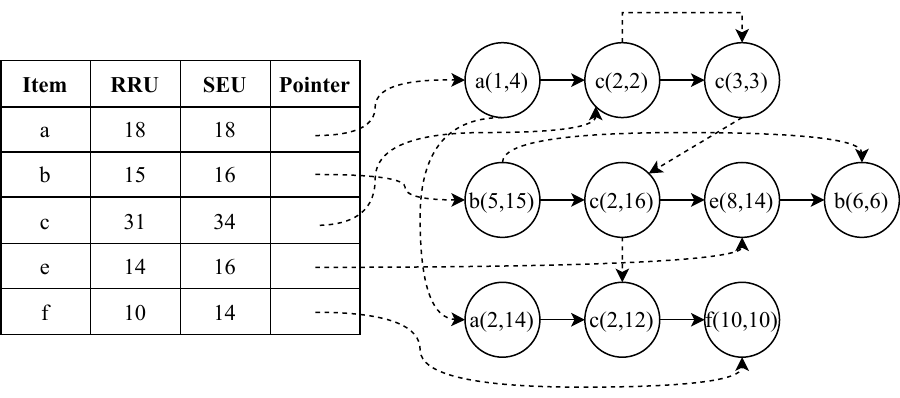}
    \caption{The utility-linked table (\textbf{ULT}) consist of $s_1$, $s_2$, and $s_3$}
    \label{ULTExample}
\end{figure}

The sequence record table (\textbf{SRT}) is used to record and generate candidate sequences. It consists of four fields: \textit{Item}, \textit{OccurrenceMap},  \textit{UntilUtility}, and \textit{RRS}. An example candidate sequence $<$$a$, $c$, $f$$>$ is shown in Table \ref{SRTExample}. The \textit{Item} field stores the name of the item. The \textit{OccurrenceMap} field contains a list of tuples representing the current occurrence positions of the sequence, where each tuple consists of a sequence SID and the position of the last item within that sequence. The \textit{UntilUtility} field stores the utility of the sequence up to the current point. The \textit{RRS} field stores the \textit{RRS} value after the most recent sequence expansion. Specifically, when the sequence contains only a single item, this value is represented by the \textit{RRU} of the item in the database. For clarity, RSC uses \textit{SRT[i]} to denote the $i$-th row in the SRT. For example, in Table \ref{SRTExample}, \textit{SRT[3]} refers to the attributes associated with the candidate sequence $<$$a$, $c$, $f$$>$. For example, in Fig. \ref{ULTExample}, sequence $<$$a$$>$  appears in $s_1$ and $s_3$, $u($$<$$a$$>$$)$ $=$ $3$ and \textit{RRU}$(a)$ $=$ 18, which form the first row in Table \ref{SRTExample}. Since the second and third rows in Table \ref{SRTExample} record the information of sequences composed of accumulated items, they represent the properties of sequences $<$$a$, $c$$>$ and $<$$a$, $c$, $f$$>$, respectively.

\begin{table}[h]
	\centering
	\caption{SRT of candidate sequence $<$$a$, $c$, $f$$>$ in Fig. \ref{ULTExample} }
	\label{SRTExample}
        \begin{tabular}{|c|c|c|c|c|}  
    		\hline 
    		 \textbf{Item} & \textbf{OccurrenceMap} &  \textbf{UntilUtility} & \textbf{RRS}  \\  \hline 
    		$a$ & $[\{s_1,1\},$ $\{s_3,1\}]$ &  3 & 18  \\  \hline 
    		$c$ & $[\{s_1,2\},$ $\{s_3,2\}]$ &  8 & 18  \\  \hline 
            $f$ & $[\{s_3,3\}]$ &  14 & 14  \\     		\hline             
    \end{tabular}
\end{table}

\subsection{Proposed algorithm}

The primary algorithm and corresponding workflow are described in this subsection. To efficiently generate rules from candidate sequences, the following theorems are presented.

\begin{theorem}
    \label{ruleProduceJudge}
    \rm A sequence $s$ $=$ $<$$e_1$, $\cdots$,  $e_n$$>$ and its corresponding SRT can generate rules which are composed of items in $s$ arranged in order only if \textit{SRT[n].sup} $\geq$ \textit{SRT[n-1].sup} $\times$ \textit{minconf}. 
\end{theorem}
\begin{proof}
    First, we prove the support monotonicity of pattern representation combinations in the SRT, namely \textit{SRT[1].sup} $\geq$ \textit{SRT[2].sup} $\geq$ $\cdots$ $\geq$ \textit{SRT[n].sup}. For \textit{SRT[n+1]} and \textit{SRT[n]}, which denote \{$e_1$, $\cdots$, $e_n$, $e_{n+1}$\} and \{$e_1$, $\cdots$, $e_n$\} respectively, since \textit{SRT[n+1]} is an extension of \textit{SRT[n]}, we need to perform further screening at each occurrence position $e_{n+1}$ of \textit{SRT[n+1]}. Specifically, we need to check whether there exists an element $e_k$ in the sequence such that $e_k$ $=$ $e_{n+1}$ and $k$ $>$ $n$ holds. Therefore, the support of \textit{SRT[n+1]} is less than or equal to that of \textit{SRT[n]}. Meanwhile, since SRT only records sequences with valid support (i.e., \textit{SRT[1].sup} $>$ $0$), \textit{SRT[1].sup} $\geq$ \textit{SRT[2].sup} $\geq$ $\cdots$ $\geq$ \textit{SRT[n].sup} holds by mathematical induction.
    According to Definitions \ref{HUSRDefinition} and \ref{supAndConf}, the rule $r$ must satisfy $sup(r)/sup(X)$ $\geq$ \textit{minconf}. Besides, the support satisfies the downward closure property, which means \textit{SRT[1].sup} $\geq$ \textit{SRT[2].sup} $\geq$ $\cdots$ $\geq$ \textit{SRT[n].sup}. For all rules that can be generated from $s$, the support \textit{sup(r)} is fixed as \textit{SRT[n].sup}. Therefore, if the minimum $sup(X)$, i.e., \textit{SRT[n-1].sup}, does not satisfy $sup(r)/sup(X)$ $\geq$ \textit{minconf}, then $s$ cannot generate any rules. 
\end{proof}

In Table \ref{SRTExample}, \textit{minconf} $=$ 0.6. When the candidate sequence $s$ $=$ $<$$a$, $c$$>$, $sup(<$$a$, $c$$>)$ $=$ $2$ $\geq$ $sup$$(<$$a$$>)$ $\times$ \textit{minconf} $=$ 1.2, therefore $s$ can generate rules. When the candidate sequence $s$ $=$ $<$$a$, $c$, $f$$>$, $sup($$<$$a$, $c$, $f$$>$$)$ $=$ $1$ $<$ $sup$$(<$$a$, $c$$>)$ $\times$ \textit{minconf} $=$ 1.2; therefore, in this case, $s$ can not generate rules.

\begin{theorem}
    \label{ruleProduceQ}
    \rm Consider a sequence $s$ $=$ $<$$e_1$, $\cdots$, $e_n$$>$, if \textit{SRT[k+1].sup} $\leq$ \textit{SRT[n].sup} $/$ \textit{minconf} $<$ \textit{SRT[k].sup}, $s$ could generate at most $n-k-1$ rules composed of items in $s$ arranged in order, where $1$ $\leq$ $k$ $\leq$ $n$. When $k$ $<$ $n$, by iteratively increasing $k$ and outputting rules in the form $r$ $=$ $<$ \textit{SRT[1].Item}, $\cdots$, \textit{SRT[k].Item}$>$ $\rightarrow$ $<$ \textit{SRT[k+1].Item}, $\cdots$, \textit{SRT[n].Item}$>$, all possible rules can be generated, ensuring the completeness of the resulting rule set.
\end{theorem}
\begin{proof}
    According to Definitions \ref{HUSRDefinition} and \ref{supAndConf}, the rule $r$ must satisfy \textit{sup(r)}/\textit{minconf} $\geq$ \textit{sup(X)}. Since the support \textit{sup(r)} is fixed as \textit{SRT[n].sup} and \textit{SRT[1].sup} $\geq$ \textit{SRT[2].sup} $\geq$ $\cdots$ $\geq$ \textit{SRT[n].sup} $=$ \textit{sup(r)}, the value on the left side of the inequality \textit{sup(r)} $\geq$ \textit{sup(X)} $*$ \textit{minconf} is known. Additionally, since \textit{minconf} is a known value, we only need to find all \textit{sup(X)} that satisfy the inequality to ensure the completeness of rule generation.
    Based on the support monotonicity of the SRT, when it is found that the support of \textit{SRT[k+1]} satisfies the inequality while the support of \textit{SRT[k]} does not, for all values of $i$ where $1<i<k$,  \textit{SRT[i].sup} $\geq$ \textit{SRT[k].sup}. Consequently, none of these satisfy the requirement of \textit{minconf}. Likewise, for all  $k$ $\leq$ $i$ $\leq$ $n$  , \textit{SRT[i].sup} $\leq$ \textit{SRT[k].sup}. Hence, the support values from $k$ to $n$ all satisfy the requirement of \textit{minconf}, and thus at most $n - k - 1$ rules can be generated.
    By finding the maximum value of \textit{sup(X)} that satisfies the inequality, the algorithm can progressively move the $\rightarrow$ rightward to generate all rules corresponding to the sequence $s$ without requiring further confidence calculations. For each sequence, the algorithm does not miss any rule, as long as the maximum \textit{sup(X)} that satisfies the constraint \textit{sup(r)}  $\geq$ \textit{sup(X)} $*$ \textit{minconf} is found. Subsequently, by virtue of support monotonicity, all qualified sequences are traversed and generated, thus guaranteeing the completeness of rule generation.
\end{proof}

For instance, consider the sequence $s$ $=$ $<$$a$, $b$, $c$, $d$$>$, the corresponding support values are 5, 4, 3, and 2, \textit{minconf} $=$ 0.6, $\textit{sup}(r)$/\textit{minconf} $=$ 3.3 $<$ 4, $s$ could generate at most one rule. By employing the pruning strategies and data structures above, we propose a rule generation method based on segmentation with confidence (namely RSC). The main pseudocode of RSC is presented in Algorithm \ref{alg:RSC}, Algorithm \ref{alg:SRTGrowth}, and Algorithm \ref{alg:ruleProduce}, respectively.

\begin{algorithm}[ht]
    \small
    \caption{RSC algorithm}
    \label{alg:RSC}
    \KwIn{$\mathcal{D}$: the database, \textit{minconf}, \textit{minutil}.}
    \KwOut{all HUSRs.}
    scan $\mathcal{D}$ to construct ULT;\\
    use ULT.\textit{SEU} to delete unpromising items in ULT;\\
    scan ULT to compute all \textit{RRU} for each item;\\
    initialize SRT $=$ $\emptyset$, candidate sequence $s$ $=$ $\emptyset$;\\
    \ForEach{ $a$ $\in$ \textit{ULT.Item}}{
        $s$ $\leftarrow$ $s$ $\cup$ $a$;\\
        SRT.add(\textit{a.OccurrenceMap}, \textit{a.UntilUtility}, \textit{RRU}$(a)$);\\
        scan ULT to get the set $B$ of all items that occurs after $a$;\\
        \ForEach{ $b$ $\in$ $B$}{
        call \textbf{SRTGrowth}$($ULT, SRT, $s$, $b$, \textit{minutil}, \textit{minconf}$)$;\\     
        }
    }
\end{algorithm}

In Algorithm \ref{alg:RSC}, RSC first scans $\mathcal{D}$ row by row to construct the ULT and uses the \textit{SEU} of each item to eliminate all irrelevant items from $\mathcal{D}$ (lines 1-2). After applying Strategy \ref{seus}, all promising items are stored in the ULT. Therefore, RSC can scan the ULT to compute the \textit{RRU} of each item in each sequence and initialize the SRT and the candidate sequence $s$ (lines 3-4). Next, for each item $a$ $\in$ \textit{ULT.Item}, RSC adds $a$ to the candidate sequence $s$. Based on the current SRT and sequence $s$, RSC computes the attributes of the extended $s$, updates them in the SRT, and identifies a new set of extension items $B$ based on the current $s$ (lines 6-8). For each $b$ $\in$ $B$, RSC will call SRTGrowth to extend the SRT and $s$ (lines 9-11).

\begin{algorithm}[ht]
    \small
    \caption{The SRTGrowth algorithm}
    \label{alg:SRTGrowth}
    \KwIn{ULT: utility-linked table of the database, SRT: a sequence record table, $s$: the candidate sequence, $b$: current candidate item, \textit{minutil}, \textit{minconf}.}
    $s$ $\leftarrow$ $s$ $\cup$ $b$;\\
    SRT.add(\textit{b.OccurrenceMap}, \textit{b.UntilUtility}, \textit{RRS}$(s,b)$);\\
    call \textbf{ruleProduce}(SRT, \textit{minutil}, \textit{minconf})\\
    scan ULT to get the set $C$ of all items that occurs after $s$;\\
    \ForEach{ $c$ $\in$ $C$}{
    \If{\textit{RRS}$(s,c)$ $\geq$ \textit{minutil}  }{
    call \textbf{SRTGrowth}$($ULT, SRT, $s$, $c$, \textit{minutil}, \textit{minconf}$)$;
    }
    }
\end{algorithm}
 
In Algorithm \ref{alg:SRTGrowth}, RSC will extend the SRT and candidate sequences. The SRTGrowth algorithm takes ULT, SRT, the candidate sequence $s$, the current candidate item $b$,  \textit{minutil}, and \textit{minconf} as input. First, SRTGrowth extends the sequence $s$ and updates the relevant attributes of the new extension item $b$ in the new sequence $s$ within the SRT (lines 1-2). Then, it calls the ruleProduce procedure to generate all corresponding rules (line 3). Similarly, SRTGrowth scans the ULT to obtain all candidate items set $C$ for the new sequence $s$ (line 4). Since the search is conducted based on $s$, this is equivalent to projecting the database. For each new candidate item $c$ $\in$ $C$, if \textit{RRS}$(s,c)$ $\geq$ \textit{minutil}, then SRTGrowth is called recursively until all possible candidate sequences have been generated (lines 5-9).

\begin{algorithm}[ht]
    \small
    \caption{The ruleProduce algorithm}
    \label{alg:ruleProduce}
    \KwIn{SRT: a sequence record table,  \textit{minutil}, \textit{minconf}, $n$ denotes the number of items in the SRT.}
    \KwOut{all HUSRs.}
    \If{ \textit{SRT[n].UntilUtility}  $\geq$  \textit{minutil} and \textit{SRT[n].sup} $\geq$ \textit{SRT[n-1].sup} $\times$  \textit{minconf}}{
    find the index $k$ such that \textit{SRT[k].sup} $\leq$ \textit{SRT[n].sup}  $/$ \textit{minconf} $\leq$ \textit{SRT[k-1].sup};\\
    \ForEach{ $k$ $<$ $n$}{
    output the $r$ $=$ $<$ \textit{SRT[1].Item}, $\cdots$, \textit{SRT[k].Item}$>$ $\rightarrow$ $<$ \textit{SRT[k+1].Item}, $\cdots$, \textit{SRT[n].Item}$>$;\\
    k++;\\
    }
    }   
\end{algorithm}

In Algorithm \ref{alg:ruleProduce}, RSC will output all HUSRs. First, according to Definition \ref{HUSRDefinition} and Theorem \ref{ruleProduceJudge}, RSC quickly returns for any SRT that has no potential to generate rules (line 1). Next, the RSC uses Theorem \ref{ruleProduceQ} to find the index $k$ that satisfies the \textit{SRT[k].sup} $\leq$ \textit{SRT[n].sup}  $/$ \textit{minconf} $\leq$ \textit{SRT[k-1].sup}, which determines the leftmost position of the arrow "$\rightarrow$" in all rules that can be generated from the current sequence $s$ corresponding to the SRT (line 2). Then, it incrementally increases $k$, which corresponds to gradually moving the arrow "$\rightarrow$" to the right within the sequence, until all valid rules are generated (lines 4-5).

\subsection{Computational complexity analysis}

\textbf{Time complexity}: Assume the dataset contains $E$ sequences, with $L$ denoting the maximum sequence length, $I$ represents the total number of unique items, and $H$ is the maximum search depth. The initial scan of the dataset to construct the ULT requires the computational and memory requirements of $O$($E\times L$). For each distinct item, the algorithm performs the depth-first generation of candidate sequences, resulting in $I$ recursive steps. Each step may at most require scanning the dataset, leading to the time complexity of $O$($E \times L$). After a candidate sequence is generated, rule generation takes $O$($\log_2I$)  time. Therefore, the complete  time complexity is: $O$($E$ $\times$ $L$ $+$ $I$ $\times$ ($E$ $\times$ $L$ $+$ $\log_2I$)). 

\textbf{Space complexity}: At each level, candidate items must be retained in memory, and each itemset records the occurrence of candidate items. Therefore, the maximum space required for candidate sequences at a single level is $I\times E$. In addition, the algorithm maintains an SRT, which reaches its maximum size at depth $H$, that is $H\times E$. Therefore, the algorithm requires a total space complexity of $O$($E$ $\times$ $L$ $+$ $H$ $\times$ $I$ $\times E$ $+$ $H$ $\times$ $E$).

\section{Experiments}   \label{sec:experiments}

A set of experiments is carried out in this section to assess the performance of RSC. Since dataset characteristics affect algorithm results, we test on both real and synthetic datasets. Because existing totally ordered rule mining algorithms do not support datasets containing repetitive items, we adopt the same preprocessing method as the most advanced algorithm TotalSR. Specifically, within each sequence, only the item with the maximum utility among duplicates is retained. TotalSR \cite{zhang2024totally}, being the leading totally ordered rule mining algorithm, is selected for comparison. Note that TotalUS is a variant of TotalSR that incorporates the pruning strategies and techniques proposed in US-Rule \cite{huang2023us} within the TotalSR framework.

\subsection{Dataset description}

Table \ref{datasets}  presents the main characteristics of the datasets. $\vert\mathcal{D}\vert$ indicates the dataset size in terms of sequences, while $\vert \textit{$\Sigma$} \vert$  corresponds to the count of unique items. The average number of items per sequence is represented as $\textit{avg}(\textit{S})$, and the maximum number of items per sequence is represented as $\textit{max}(\textit{S})$. All datasets are available on the open-source platform SPMF \cite{spmf} and have been widely used in evaluations of related algorithms \cite{2023Anomaly,zhang2024totally,lin2023user,10909343}. The descriptions of each dataset are provided below.

\begin{table}[ht]
	\caption{Features of the datasets} 
	\label{datasets}
	\centering
	\begin{tabular}{|c|c|c|c|c|}
    	\hline
    	\textbf{Dataset} &  \textbf{$\vert\mathcal{D}\vert$}  &  \textbf{$\vert \textit{$\Sigma$} \vert$} & $\textit{avg}(\textit{S})$ & $\textit{max}(\textit{S})$  \\
    	\hline 
            Bible & 36,369 & 13,905 & 21.64 & 100  \\ \hline
            Kosarak10k & 10,000 & 10,094 & 8.14 & 608 \\ \hline
            Leviathan & 5,834 & 9,025 & 33.81 & 100  \\ \hline
            SIGN & 730 & 267 & 52 & 94 \\ \hline
    	Syn10K  & 10,000  & 7312 & 27.11 & 213  \\ \hline
            Syn20K  & 20,000  & 7442 & 26.98 & 213\\ \hline
            Syn40K  & 40,000  & 7537 & 26.84 & 213  \\ \hline
            Syn80K  & 79718 & 7584 & 26.79 & 213\\ \hline
            Syn160K  & 159501  & 7609 & 26.75 & 213  \\ \hline
            Syn240K  & 239211  & 7617 & 26.77 & 213\\ \hline
            Syn320K  & 318889  & 7620 & 26.76 & 213  \\ \hline
            Syn400K  & 398716  & 7621 & 26.75 & 213\\
        \hline 
	\end{tabular}
\end{table}

\begin{figure*}[h]
    \centering
    \includegraphics[width=1\textwidth]{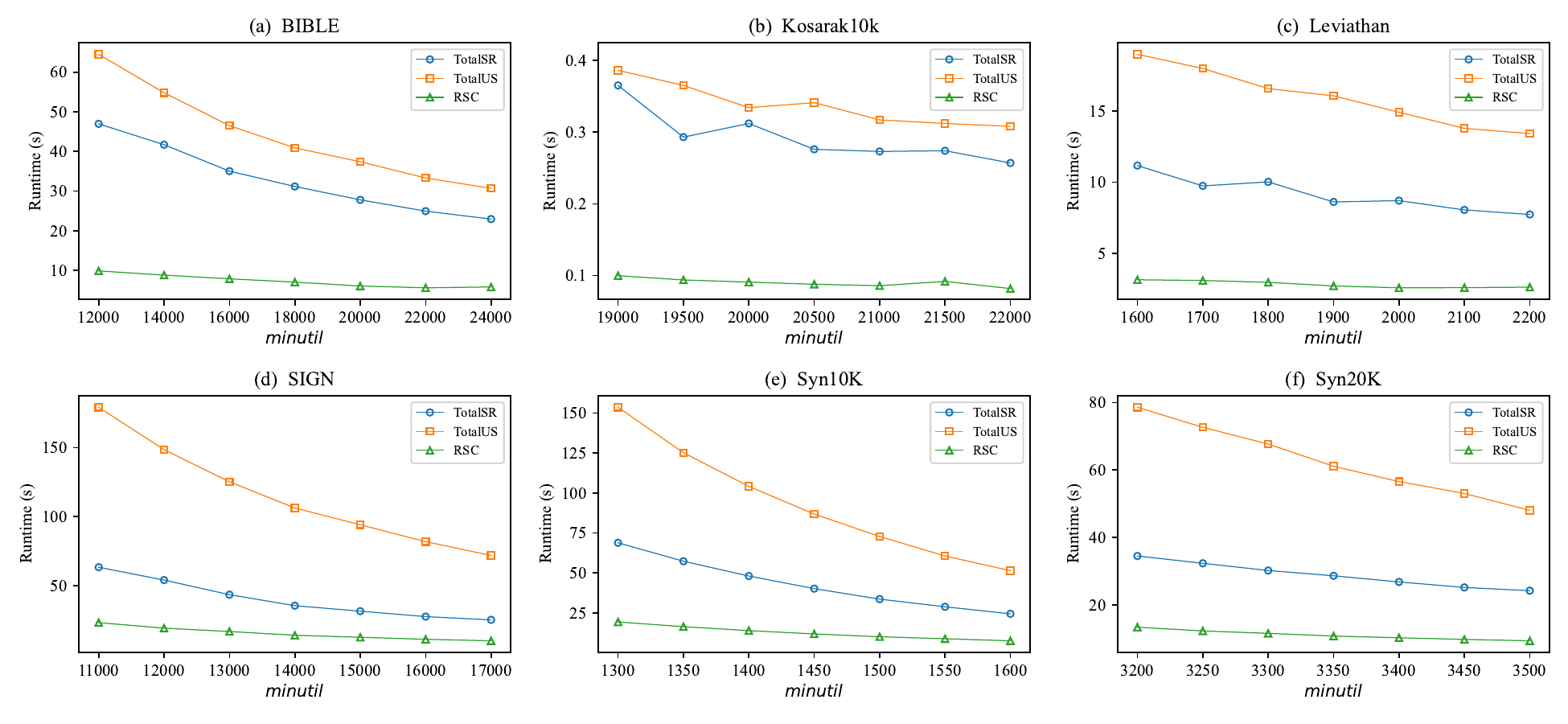}
    \caption{Runtime under different \textit{minutil} (\textit{minconf} $=$ 0.6) }
    \label{TimeTest}
\end{figure*}

\begin{itemize}[leftmargin=*,labelsep=0.5em,align=left]
    \item \textbf{Bible} dataset is derived from the book \textit{BIBLE}. Each sequence represents a Bible sentence and has moderate length and medium density.
    \item \textbf{Kosarak10k} consists of real clickstream data, with short sequences collected from a Hungarian online news platform. It is characterized by sparse distribution and brief sequence lengths.
    \item \textbf{Sign} dataset contains real-world sequences captured from sign language demonstrations. It is highly dense and contains many long sequences. 
    \item \textbf{Leviathan} is based on textual content from Thomas Hobbes’ novel \textit{Leviathan}. It has moderate density and medium-length sequences.
    \item The \textbf{Syn} datasets are artificially generated using the IBM data generator. The primary difference among them lies in the dataset size.
\end{itemize}

\subsection{Parameter settings}
To ensure evaluation reliability, all experiments were conducted on a PC equipped with Ubuntu 24.04.2 with the 6.11.0-17 Linux kernel, equipped with a 12th Gen Intel® CoreTM i9-12900K processor and 64 GB of RAM. In the experiments, the \textit{minconf} for all algorithms is set to the default value of $0.6$. The \textit{minutil} of all algorithms is set to different values,  as shown in Fig. \ref{TimeTest}. All experimental implementations were realized in Java, and the source code alongside datasets is provided on GitHub: \href{https://github.com/IOTS-HIT/Segmentation}{https://github.com/IOTS-HIT/Segmentation}.

\begin{figure*}[h]
    \centering
    \includegraphics[width=1\textwidth]{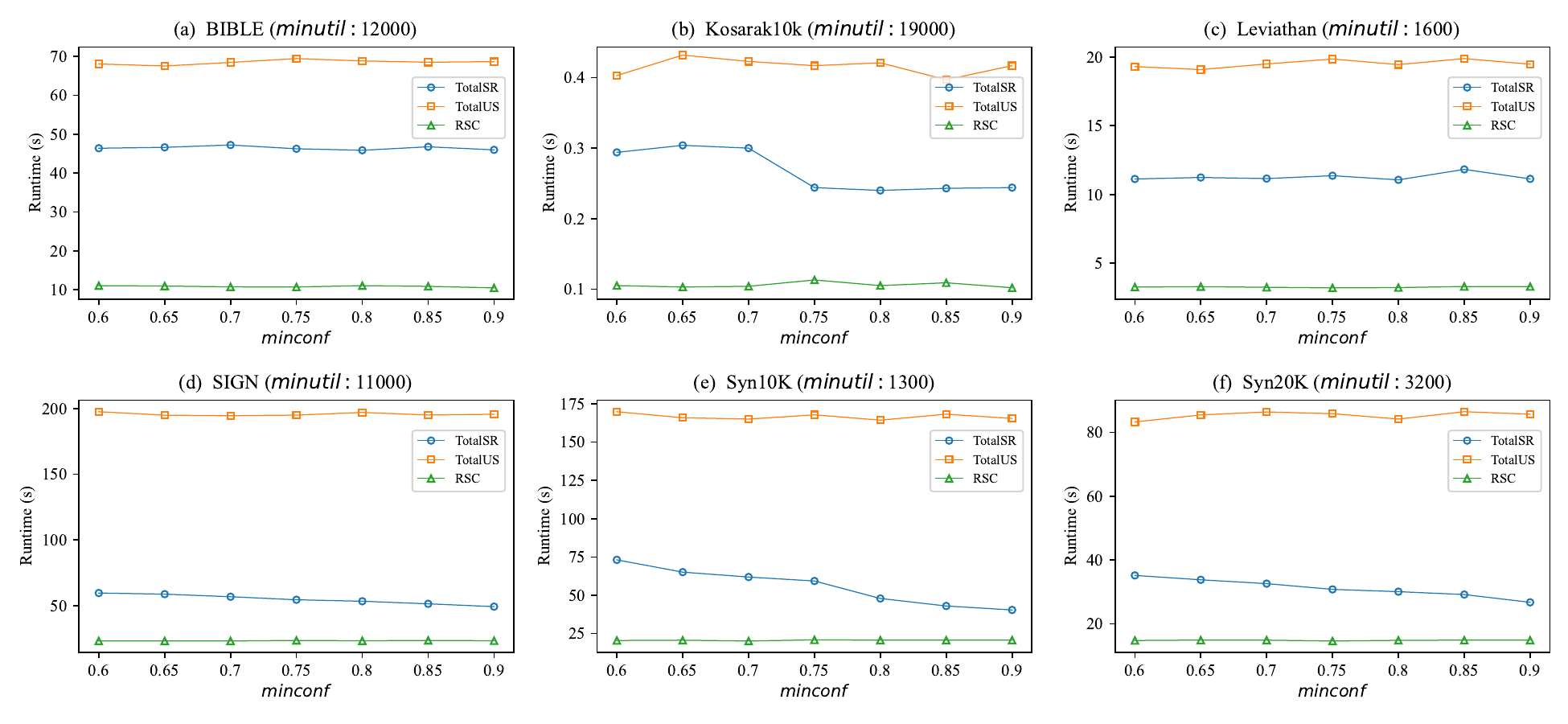}
    \caption{Runtime under different \textit{minconf} }
    \label{ConfTest}
\end{figure*}

\begin{figure*}[h]
    \centering
    \includegraphics[width=1\textwidth]{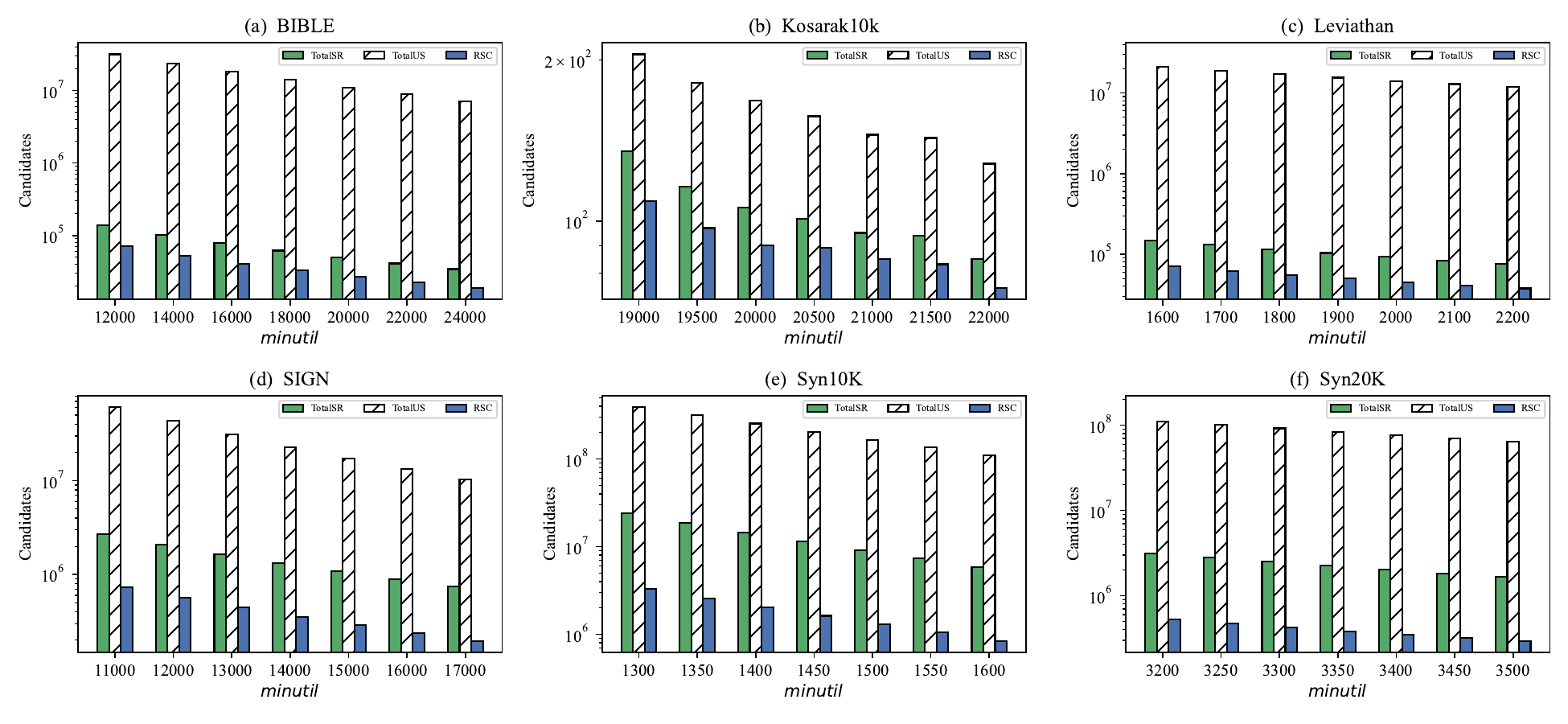}
    \caption{The generated candidates under different \textit{minutil} (\textit{minconf} $=$ 0.6)}
    \label{TableTest}
\end{figure*}

\begin{figure*}[t]
    \centering
    \includegraphics[width=1\textwidth]{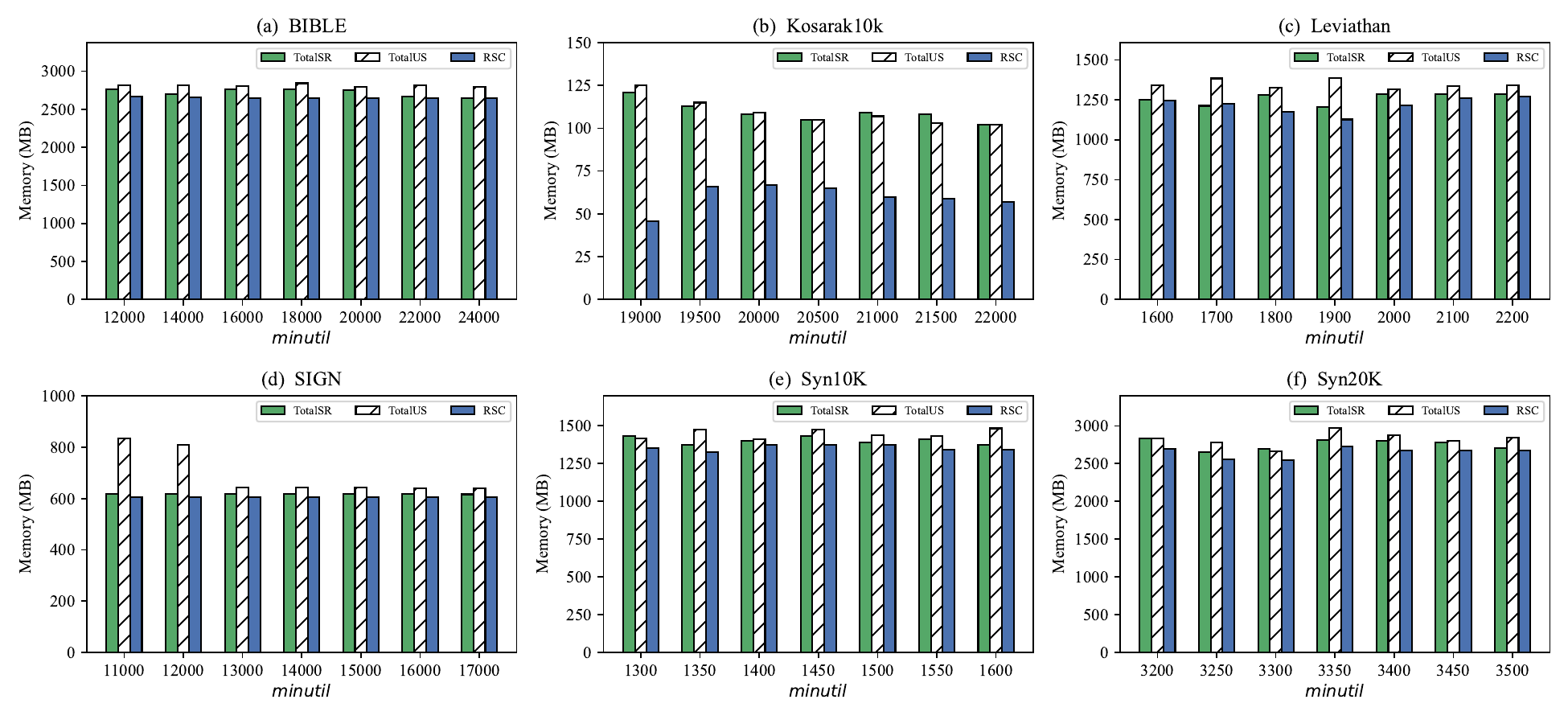}
    \caption{The memory usage under various \textit{minutil} (\textit{minconf} $=$ 0.6) }

    \label{SizeTest}
\end{figure*}

\begin{figure*}[t]
    \flushright
    \includegraphics[width=1\textwidth]{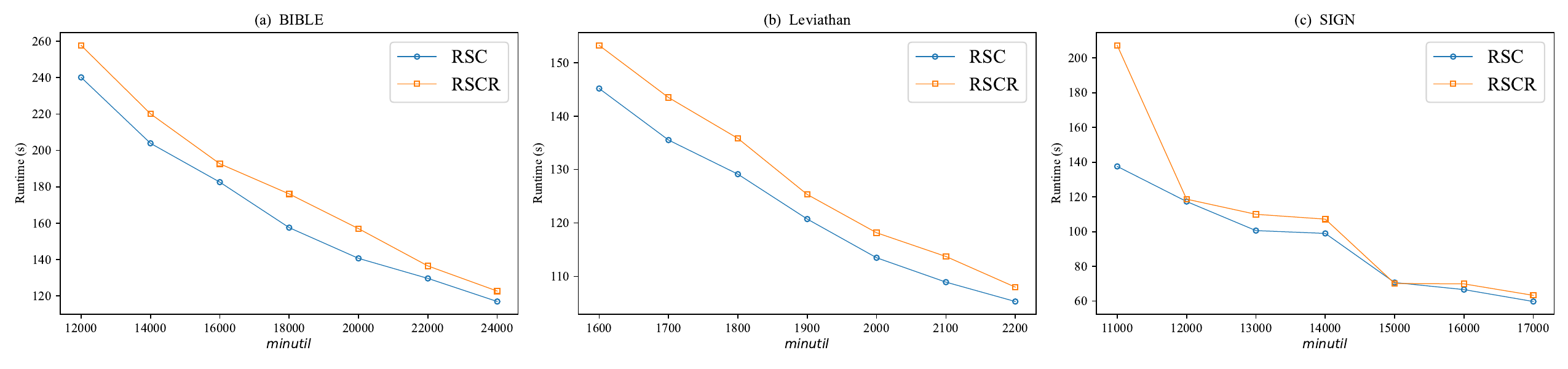}
    \caption{Runtime  under different upper bounds (\textit{minconf} $=$ 0.6) }
    \label{timeP}
\end{figure*}

\subsection{Efficiency analysis}

The runtime of each algorithm under different \textit{minutil} thresholds is presented in Fig. \ref{TimeTest}, while Fig. \ref{TableTest} presents the total number of candidate sequences generated during execution. As shown in Fig. \ref{TimeTest}, the runtime of all methods decreases as the threshold increases. This observation aligns with expectations, as a higher threshold results in fewer rules satisfying the utility constraint. Moreover, due to the integrated pruning strategies, the algorithms can effectively backtrack from the current search depth once further expansion is deemed unnecessary. Additionally, the RSC algorithm consistently achieves the shortest runtime. This is expected, as RSC avoids redundant utility computations, significantly reducing the number of candidate sequences and thereby lowering the overall execution time. This effect is particularly evident on larger datasets such as BIBLE, SIGN, Syn10K, and Syn20K, where RSC achieves a runtime reduction of approximately 50\% to 80\% compared to existing methods such as TotalSR and TotalUS. This significant improvement is closely related to RSC’s ability to avoid redundant utility computations. Additionally, the runtime is positively correlated with the number of generated candidate sequences. In Fig. \ref{TableTest}, candidate generation by each algorithm decreases as the threshold is lowered. Compared to TotalUS, TotalSR applies tighter utility upper bounds, which reduce the search space and improve efficiency. However, due to redundant computations, it still generates more candidate sequences and requires more runtime than RSC. 

Fig. \ref{ConfTest} illustrates the impact of different \textit{minconf} settings on runtime under a specific \textit{minutil} value. Specifically, RSC consistently ranks first in speed, followed by TotalSR in second place and TotalUS in third, which is consistent with the previous experimental results. It is worth noting that the runtime of the three algorithms does not change significantly with the variation of confidence. This is because RSC only needs to determine the cut point position based on confidence without adding additional calculations, while TotalSR incorporates confidence-based pruning compared to TotalUS, resulting in a slight decrease in computation time. As can be seen from the figure, when a certain level of confidence is considered, the setting of the confidence level is unlikely to become a limiting factor for runtime.
\vfill

\begin{figure*}[t]
    \flushright
    \includegraphics[width=1\textwidth]{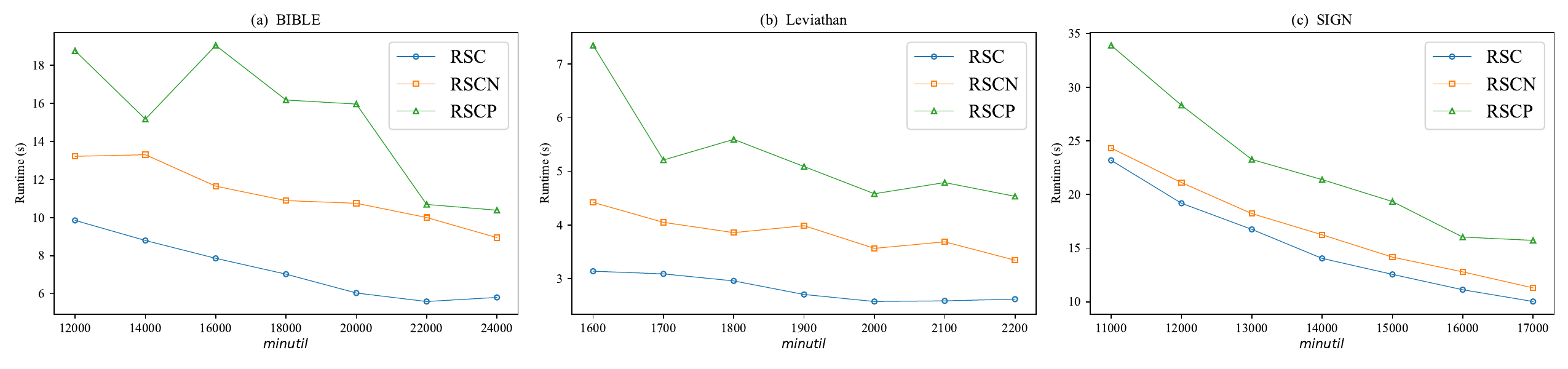}
    \caption{Runtime under different strategies (\textit{minconf} $=$ 0.6) }
    \label{timeS}
\end{figure*}

\begin{figure*}[h]
    \centering
    \includegraphics[width=1\textwidth]{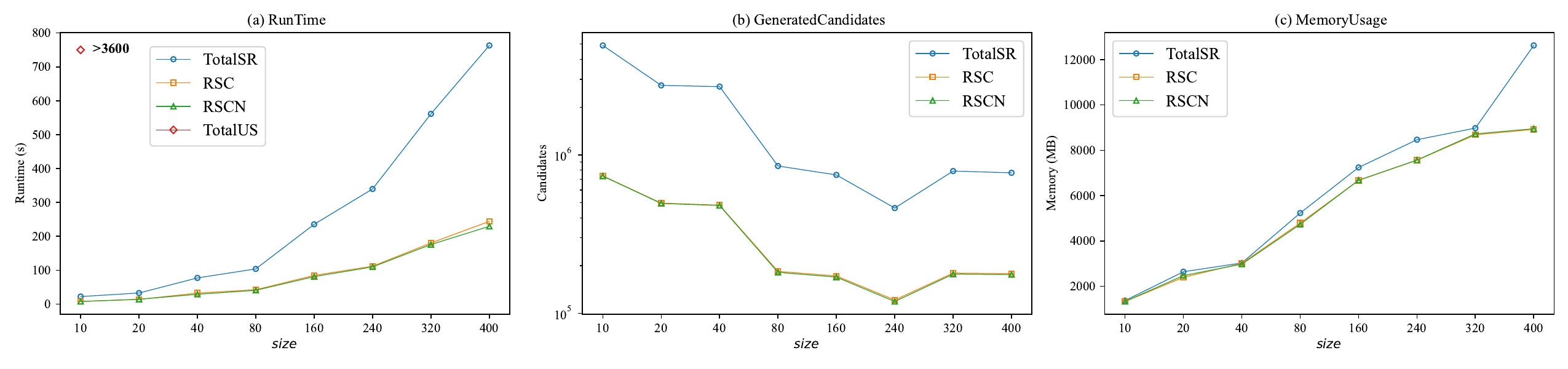}
    \caption{Scalability test ($\delta$ $=$ $1\%$ and \textit{minconf} $=$ 0.6)}  
    \label{sacTest}
\end{figure*}

\subsection{Memory analysis}

The memory consumption of each algorithm during execution is shown in Fig. \ref{SizeTest}. As illustrated, denser datasets and longer sequence lengths require more memory during execution. The memory usage across algorithms is generally similar, as the search depth required to mine all HUSRs is the same. However, the design of data structures and pruning strategies can slightly reduce memory usage. TotalSR reduces memory consumption by improving the upper bounds of TotalUS, thereby reducing the search space depth. As shown in Fig. \ref{TableTest}, RSC further improves memory usage because it adopts a single-item extension strategy, which avoids redundant item expansions and reduces the number of candidate items stored. In Fig. \ref{SizeTest}(b), RSC shows substantially lower memory usage than the other methods. This is likely due to the shorter average sequence lengths in the Kosarak10k dataset, which leads to a shallower search space. Under such conditions, RSC’s ability to generate all rules from a candidate sequence in a single step effectively reduces both the number of candidate sequences and the associated memory cost.

\subsection{Ablation experiments}

We designed three RSC variants to assess the effectiveness of the proposed strategies: RSCN, RSCP, and RSCR. RSCN disables Strategy \ref{seus} but retains all other pruning mechanisms. RSCP disables Strategy \ref{seps}. RSCR replaces the RRU upper bound with the RU upper bound. Since RRU is an improvement on RU, RSC can only demonstrate its high efficiency in sequences containing duplicate items. For this reason, the experiment for Fig. \ref{timeP} was conducted on the original dataset with duplicate items, while that for Fig. \ref{timeS} was performed on the dataset with additional preprocessing to remove duplicate items. Although the parameter settings are identical for both experiments, their performance metrics (e.g., time efficiency) exhibit notable differences.

Fig. \ref{timeP} shows the runtime of RSC and RSCR. As shown, the runtime of RSC is significantly lower than that of RSCR on BIBLE and Leviathan. This is because RSC uses a tighter utility upper bound, which allows it to prune the search space more effectively, avoid unnecessary search operations, and reduce the runtime. The dense structure and relatively long sequences in these datasets amplify the benefits of the tighter bound, as its computational overhead is low compared to the search space reduction it enables. The performance difference between the two algorithms on SIGN is small. This is because sequences in SIGN are generally long, resulting in a higher average remaining utility for each item. As a result, the estimated utility upper bounds for both algorithms increase, making it easier to exceed \textit{minutil} without pruning and thus reducing the distinction between the two upper bounds.

Fig. \ref{timeS} shows the runtime of RSC, RSCP, and RSCN, where RSCN does not use Strategy \ref{seus}, and RSCP does not adopt Strategy \ref{seps}. According to Fig. \ref{TimeTest}, RSCN consistently achieves the second-best runtime, as it employs the same rule generation method as RSC. However, since RSCN does not utilize Strategy \ref{seus},  it fails to eliminate a large number of irrelevant items that do not contribute to any rule generation. As a result, it exhibits slower performance than RSC. Furthermore, RSCP consistently exhibits the longest runtime. This is because it does not adopt Strategy \ref{seps}, a critical pruning strategy. Without this strategy, RSCP cannot predict in advance which extensions are impossible to form valid rules, thus requiring additional time for in-depth downward search. The runtime difference is small on SIGN, likely due to the extremely long sequence lengths in this dataset. Such long sequences inflate the SEU estimates and weaken the pruning effect of Strategy \ref{seus}. Additionally, Strategy \ref{seps} struggles to perform more effective pruning due to excessively high remaining utility, thereby reducing the performance gap.

\subsection{Case study}

The aforementioned experiments evaluated the performance of the RSC algorithm across various dimensions. To fully demonstrate the effectiveness of this algorithm, we analyzed the rules derived from the Kosarak10k dataset. In this dataset, each item represents a news category in the user clickstream, and utility is assumed to correspond to the advertising traffic revenue generated. With \textit{minutil} set to 22000 and \textit{minconf} set to 0.6, we can obtain 14 high-utility rules, such as $\{11, 218,6 \rightarrow 148\}$ and $\{11,27 \rightarrow 6\}$. Particularly, the item $6$ is involved in the formation of all 14 rules. Analysis of these rules enables us to obtain high-value records of user preferences, identify the types of news that users prefer to browse, and thus optimize recommendation configurations. When partially ordered rules are applied to recommendation tasks, the core objective tends to be utility maximization. However, since the internal order of antecedent and consequent is undefined, the recommended items may be irrelevant to the news categories that users are currently browsing. In contrast, totally ordered rules allow recommendations to be generated based on users’ real-time browsing records. For example, when a user browses news about a specific brand of car, totally ordered rules can generate an optimal recommendation sequence—such as first recommending other types of cars from the current brand and then gradually guiding the user to recommendations of cars from other brands. In contrast, partially ordered rules can only provide a set of recommendations, and the order in which these recommendations should be presented is undefined. Consequently, this may result in low interest from the current user.  Furthermore, the news category represented by the item $6$ is a focal point of user attention. It can serve as an excellent traffic-driving or advertising placement channel and thus should be featured on the website homepage with higher advertising rates applied.

\subsection{Scalability test}

Fig. \ref{sacTest} presents the performance of algorithms on larger datasets to test their robustness. The datasets are synthesized from the real datasets mentioned earlier and vary in size from 10,000 to 398,716. The confidence is set to 0.6, and the $\delta$ is set to $1\%$. Since the TotalUS algorithm could hardly complete any experiments, we did not conduct all its experiments. Specifically, its runtime on Syn10K has already exceeded 3600 seconds, which is far longer than that of other algorithms. This is because the selected low threshold, along with the long sequence lengths and large number of sequences in the dataset, has greatly amplified the deficiencies of its pruning strategy, leading to excessive computational overhead and failure to complete subsequent experiments within a reasonable time frame. In Fig. \ref{sacTest}(a), the running time of each algorithm increases with the dataset size, and the time required by TotalSR is significantly higher than the other algorithms, which is consistent with previous experiments. RSC and RSCN yield comparable experimental results because the datasets contain many items that appear repeatedly in multiple sequences, which greatly prevents them from being regarded as irrelevant by Strategy \ref{seus} and deleted early. Therefore, the difference between RSC and RSCN is small. In Fig. \ref{sacTest}(b), the number of candidate items is not linearly correlated with the dataset size. This is because the \textit{minutil} is determined by a relative ratio $\delta$, and the datasets differ in structure and distribution. As a result, the \textit{minutil} derived from the same $\delta$ value differs across datasets, leading to variations in the size and level of search space exploration. As a result, the count of candidate items does not exhibit a strictly positive correlation with dataset size. In Fig. \ref{sacTest}(c), the memory usage of all algorithms grows as the dataset size increases. RSC requires less memory compared to TotalSR, demonstrating its superior scalability. In summary, the RSC algorithm efficiently handles large-scale datasets and achieves improved performance compared to the other three algorithms.

\section{Conclusion}   \label{sec:conclusion}

In this paper, we identify the issue of redundant utility computation in existing HUSRM algorithms and propose the RSC algorithm to address it, based on the observation that different rules often share the same underlying item sequence. RSC employs a segmentation-based approach to generate all distinct rules from the candidate sequences, significantly reducing redundant utility computations. Moreover, RSC introduces a novel projected data structure—the utility-linked table—to accelerate candidate generation. It also proposes a tighter utility upper bound tailored for sequences containing repetitive items. Extensive evaluations using real and synthetic data indicate that RSC consistently outperforms advanced algorithms, including TotalSR. A valuable avenue for future work is to explore high-utility sequential rule mining under various types of constraints, such as closed constraints, gap constraints, or target constraints. Incorporating such constraints can help generate more actionable and valuable rules that better align with the needs of real-world applications. Furthermore, the proposed RSC algorithm has the potential to be extended in multiple ways. One promising line of work is to adapt it for more complex datasets; another is to redesign RSC into a parallel or distributed framework, which would markedly increase scalability, facilitating its use in real-time and large-scale data mining applications.

\ifCLASSOPTIONcompsoc
  \section*{Acknowledgments}
\else
  \section*{Acknowledgment}
\fi

This work was supported in part by the Natural Science Foundation of Guangdong Province (No. 2024A1515010242), National Natural Science Foundation of China (No. 62272196), and Shenzhen Research Council (No. GXWD202208111702530022).

\bibliographystyle{IEEEtran}
\bibliography{RSC.bib}

\end{document}